\documentclass{amsart}
\usepackage[top=4cm,bottom=4cm,left=3.5cm,right=3.5cm]{geometry}

\usepackage{amsmath,amssymb}
\usepackage[english]{babel}
\usepackage{mathtools}
\usepackage{enumerate}
\usepackage[hidelinks]{hyperref}
\usepackage{tikz}
\usepackage{tkz-graph}
\usepackage{tikz-cd}
\usepackage{bbold}
\usepackage{multirow}
\usepackage{todonotes,listings}
\usepackage{xcolor}
\usepackage{algorithm}
\usepackage[noend]{algpseudocode}

\usepackage{graphicx}
\usepackage{epstopdf}

\newcommand{\vF}{\mathbb{F}}

\newcommand{\cC}{\mathcal{C}}

\newcommand{\hhs}{\widehat{H}_S}
\newcommand{\ccs}{\widehat{\cC}_S}

\newtheorem{theorem}{Theorem}[section]

\newtheorem{lemma}[theorem]{Lemma}
\newtheorem{proposition}{Proposition}

\theoremstyle{definition}
\newtheorem{definition}[theorem]{Definition}
\newtheorem{remark}{Remark}

\def\cprime{$'$}
\hypersetup{urlcolor=blue, citecolor=red}
\numberwithin{equation}{section}
\begin{document}

\title[Encryption Scheme Based on Expanded Reed-Solomon Codes]{Encryption Scheme Based on Expanded Reed-Solomon Codes}

\author[K. Khathuria]{Karan Khathuria}
\address{Institute of Mathematics\\
University of Zurich\\
Winterthurerstrasse 190\\
8057 Zurich, Switzerland\\
}
\email{karan.khathuria@math.uzh.ch}

\author[J. Rosenthal]{Joachim Rosenthal}
\address{Institute of Mathematics\\
University of Zurich\\
Winterthurerstrasse 190\\
8057 Zurich, Switzerland\\
}
\email{rosenthal@math.uzh.ch}

\author[V. Weger]{Violetta Weger}
\address{Institute of Mathematics\\
University of Zurich\\
Winterthurerstrasse 190\\
8057 Zurich, Switzerland\\
}
\email{violetta.weger@math.uzh.ch}

%
%
%
%
%
%
%
%


\begin{abstract}
We present a code-based public-key cryptosystem, in which we use Reed-Solomon codes over an extension field as secret codes and disguise it by considering its shortened expanded code over the base field. Considering shortened expanded codes provides a safeguard against distinguisher attacks based on the Schur product. Moreover, without using a cyclic or a quasi-cyclic structure we obtain a key size reduction of nearly $45 \%$ compared to the classic {McE}liece cryptosystem proposed by Bernstein \textit{et al.}

\end{abstract}

\maketitle

\section{Introduction}
\label{sec:introduction}
In 1978 McEliece \cite{mc78} presented the first code-based public key cryptosystem. It belongs to the family of very few public-key cryptosystems which are unbroken since decades. The hard problem the McEliece system relies on, is the difficulty of decoding a random (-like) linear code having no visible structure. McEliece proposed to use binary Goppa codes for the encryption scheme. Due to the low error-correcting capacity of Goppa codes, the cryptosystem results in large public key sizes. Several alternative families of codes have been proposed with the aim of reducing the key sizes. Some of the famous families of codes considered are: generalized Reed-Solomon codes \cite{ba11,ba15z,ba19,be05a,bo16,kh18,ni86a}, non-binary Goppa codes \cite{be10}, algebraic geometric codes \cite{ja96}, LDPC and MDPC codes \cite{ba08p,mi13}, Reed-Muller codes \cite{si94b} and convolutional codes \cite{lo12}. Most of them were unsuccessful in hiding the structure of the private code \cite{co14,co17,co17w,co15a,la13,mi07,ot10,si92,wi10}.

The motivation to quest for better code-based cryptosystems is mainly due to the advent of quantum computers. In 1994 Peter Shor \cite{sh94} developed a polynomial time quantum algorithm for factoring integers and solving discrete logarithm problems. This means that most of the currently popular cryptosystems, such as RSA and ECC, will be broken in an era of quantum computers. In the ongoing process of the standardization of quantum-resistant public-key cryptographic algorithms by the National Institute of Standards and Technology (NIST), code-based cryptosystems are one of the most promising candidates.  At the time of this writing there are seven code-based cryptosystems included in NIST's standardization process: BIKE \cite{ar17} based on quasi-cyclic MDPC codes, classic McEliece \cite{be08} based on binary Goppa codes, ROLLO \cite{ca18} based on quasi-cyclic LRPC codes, RQC \cite{ag18} based on rank metric quasi-cyclic codes, HQC \cite{ag18} based on Hamming metric quasi-cyclic codes, LEDAcrypt \cite{ba18b} based on quasi-cyclic LDPC codes and NTS-KEM \cite{al18} based on binary Goppa codes. 

In this paper we present a new variant of the McEliece scheme using expanded Reed-Solomon codes. A linear $[n,k]$ code defined over an extension field $\vF_{q^m}$ can be expanded, over the base field $\vF_q$, to a $[mn,mk]$ linear code by expanding each codeword with respect to a fixed $\vF_q$-linear isomorphism from $\vF_{q^m}$ to $\vF_q^m$. In the proposed cryptosystem we hide the structure of an expanded GRS code by puncturing and permuting the columns of its parity check matrix and multiplying by an invertible block diagonal matrix. In order to decode a large number of non-codewords, we use a burst of errors during the encryption step, i.e. we consider error vectors having support in sub-vectors of size $\lambda$. This error pattern comes with a disadvantage: it can be used to speed up the information set decoding (ISD) algorithms. However, for a small degree of extension $m$, the key sizes turn out to be remarkably competitive.

The paper is organized as follows. In Section \ref{sec:background}, we give the preliminaries regarding the expanded codes. In Section \ref{sec:cryptosystem}, we describe the proposed cryptosystem which is based on the shortening of an expanded generalized Reed-Solomon code. In Section \ref{sec:Security}, we provide security arguments for the proposed cryptosystem against the known structural and non-structural attacks. In Section \ref{sec:KeySize}, we provide parameters of the proposed cryptosystem that achieve a security level of 256-bits against the ISD algorithm.

\section{Background} \label{sec:background}

\subsection{Expanded Codes}

Let $q$ be a prime power and let $m$ be an integer. 
Let $\gamma$ be a primitive element of the field $\vF_{q^m}$, i.e. $\vF_{q^m} \cong \vF_q(\gamma) $. The field $\vF_{q^m}$ can also be seen as an $\vF_q$- vector space of dimension $m$ via the following $\vF_q$-linear isomorphism 
\begin{align*}
\phi: \vF_{q^m} & \longrightarrow   \vF_{q}^{m} ,\\ 
a_0 + a_1 \gamma + \cdots + a_{m-1} \gamma^{m-1}   & \longmapsto     (a_0,a_1, \ldots, a_{m-1}).
\end{align*} 
We extend this isomorphism for vectors over $\vF_{q^m}$ in the following way:
\begin{align*}
\phi_n: \vF_{q^m}^n  & \longrightarrow    \vF_{q}^{mn} , \\ 
(\alpha_0,\alpha_1,\ldots, \alpha_{n-1})  & \longmapsto   \left (\phi(\alpha_0), \phi(\alpha_1),\ldots, \phi(\alpha_{n-1}) \right ).
\end{align*}
This is clearly an $\vF_q$-linear isomorphism. Hence this gives us a way to obtain a linear code over $\vF_q$ from a linear code over $\vF_{q^m}$.
\begin{definition}[Expanded Codes]
Let $n,k$ be positive integers with $k \leq n$, let $q$ be a prime power and $m$ be an integer. Let $\cC$ be a linear code of length $n$ and dimension $k$ over $\vF_{q^m}$. The expanded code of $\cC$ with respect to a primitive element $\gamma \in \vF_{q^m}$ is a linear code over the base field $\vF_q$ defined as  \[\widehat{\cC} := \lbrace \phi_n(c) : c \in \cC \rbrace, \] where $\phi_n$ is the $\vF_q$-linear isomorphism defined by $\gamma$ as above.
\end{definition}
\begin{remark}
 It is easy to see that the expanded code $\widehat{\cC}$ is a linear code of length $mn$ and dimension $mk$, because $\phi_n$ is an $\vF_q$-linear isomorphism and\\ $|\widehat{\cC}| = |\cC| = (q^m)^k = q^{mk}$.
\end{remark}

Given a code $\cC$ with its generator matrix and parity check matrix, the following lemma gives a way to construct a generator matrix and a parity check matrix of the expanded code $\widehat{\cC}$.
\begin{lemma} Let $\cC$ be a linear code in $\vF_{q^m}^n$.
\begin{enumerate} 
\item Let $\cC$ have a generator matrix $G = [g_1, g_2,\ldots,g_k]^\intercal$, where $g_1,g_2,\ldots,g_k$ are vectors in $\vF_{q^m}^n$. Then the expanded code of $\cC$ over $\vF_q$ with respect to a primitive element $\gamma \in \vF_{q^m}$ has the expanded generator matrix \begin{align*}
\widehat{G} :=  [\phi_n(g_1),\phi_n(\gamma g_1),\ldots, \phi_n(\gamma^{m-1}g_1),& \phi_n(g_2), \phi_n(\gamma g_2),\ldots, \phi_n(\gamma^{m-1}g_2),\ldots,\\ & \phi_n(g_k),\phi_n(\gamma g_k)\ldots, \phi_n(\gamma^{m-1}g_k)]^\intercal .
\end{align*} 
\item  Let $\cC$ have a parity check matrix $H = [h_1^\intercal,h_2^\intercal,\ldots,h_n^\intercal]$, where $h_1,h_2,\ldots,h_n$ are vectors in $\vF_{q^m}^{n-k}$. Then the expanded code of $\cC$ over $\vF_q$ with respect to a primitive element $\gamma\in \vF_{q^m}$ has the expanded parity check matrix \begin{align*}
\widehat{H} :=  [ & \phi_{n-k}(h_1)^\intercal,\phi_{n-k}(\gamma h_1)^\intercal,\ldots, \phi_{n-k}(\gamma^{m-1}h_1)^\intercal, \phi_{n-k}(h_2)^\intercal,\phi_{n-k}(\gamma h_2)^\intercal,\\ & \ldots, \phi_{n-k}(\gamma^{m-1}h_2)^\intercal, \ldots, \phi_{n-k}(h_n)^\intercal,\phi_{n-k}(\gamma h_n)^\intercal\ldots, \phi_{n-k}(\gamma^{m-1}h_n)^\intercal]. 
\end{align*} 
\end{enumerate} \label{lemma:expand}
\end{lemma}
\begin{proof}
See \cite[Theorem 1]{yi11}.
\end{proof}

\begin{proposition}
Let $\cC$ be a linear code in $\vF_{q^m}^n$ having a generator matrix $G = [g_1, g_2,\ldots,g_k]^\intercal$ and a parity check matrix $H = [h_1^\intercal,h_2^\intercal,\ldots,h_n^\intercal]$. Let $\widehat{G}$ and $\widehat{H}$ be the expanded generator matrix and expanded parity check matrix of $\widehat{\cC}$, respectively. Then
\begin{enumerate}
\item $\phi_n(xG) = \phi_k(x)\widehat{G}$ for all $x \in \vF_{q^m}^k$,
\item $\phi_{n-k}(Hy^\intercal) = \widehat{H} (\phi_n(y))^\intercal$ for all $y \in \vF_{q^m}^n$.
\end{enumerate}  \label{prop:phipsi}
\end{proposition}
\begin{proof}
Let $x = (x_1,x_2,\ldots,x_k) \in \vF_{q^m}^k$ and let $x_i = \sum_{j=0}^{m-1}x_{ij}\gamma^j$ for all\\ $i \in \{1,2,\ldots,k\}$. Then 
\begin{align*}
 \phi_k(x) \widehat{G}  &=  \sum_{i=1}^{k} \sum_{j=0}^{m-1} x_{ij} \phi_n(\gamma^jg_i) \\
  & =  \sum_{i=1}^{k} \phi_n \left( \sum_{j=0}^{m-1} x_{ij} \gamma^j g_i \right) \\
  & =  \sum_{i=1}^{k} \phi_n(x_i g_i) \\
  & = \phi_n \left( \sum_{i=1}^{k} x_i g_i \right) \\
  & =  \phi_n(xG).
\end{align*} Similarly, $\phi_{n-k}(Hy^\intercal) = \widehat{H} (\phi_n(y))^\intercal$ for all $y \in \vF_{q^m}^n$.
\end{proof}
\begin{remark}
 $\widehat{\cC}$ can also be determined by the commutativity of the following diagram (as $\vF_q$-linear maps):
\[ \begin{tikzcd}
0 \arrow{r}{} & \vF_{q^m}^k \arrow{r}{G} \arrow[swap]{d}{\phi_k} & \vF_{q^m}^n \arrow{d}{\phi_n} \arrow{r}{H^\intercal} & \arrow{d}{\phi_{n-k}} \vF_{q^m}^{n-k} \arrow{r}{} & 0 \\%
0 \arrow{r}{} & \vF_q^{mk} \arrow{r}{\widehat{G}}& \vF_q^{mn} \arrow{r}{\widehat{H}^\intercal} & \vF_q^{m(n-k)} \arrow{r} & 0
\end{tikzcd}
\]
\end{remark}

\section{The Cryptosystem}\label{sec:cryptosystem}

In this section we will present the proposed cryptosystem in the Niederreiter version.  
We consider an expanded GRS code whose parity check matrix can be viewed as $n$ blocks, where each block is of size $m$. In order to destroy the algebraic structure of the code, we choose $2 \leq \lambda \leq m-1$ and shorten it on randomly chosen $m-\lambda$ columns in each block. We then hide the shortened code by multiplying it with an invertible matrix, which preserves the weight of a vector over the extension field $\vF_{q^m}$.

\paragraph{\textbf{Key generation:}}
Let $q$ be a prime power, $2 \leq \lambda < m$ be positive integers and $k< n \leq q^m$ be positive integers, satisfying $R:= k/n > (1-\lambda/m)$. 
Consider a GRS code $\cC=\text{GRS}_{n,k}(\alpha,\beta)$  of dimension $k$ and length $n$ over the finite field $\vF_{q^m}$ and choose a parity check matrix $H$ of $\cC$. Let $t$ be the error correction capacity of $\cC$. 

Let $\widehat{H}$ be the expanded parity check matrix of the expanded code $\widehat{\cC}$ of $\cC$ with respect to a primitive element $\gamma \in \vF_{q^m}$. $\widehat{H}$ is an $m(n-k) \times mn$ matrix over $\vF_{q}$.
\paragraph{Shortening $\widehat{\cC}$}
\begin{itemize} 
\item For each $1\leq i \leq n$, let $S_i$ be a randomly chosen subset of\\ $\{ (i-1)m+1,(i-1)m+2,\ldots,im \}$ of size $m-\lambda$ and define $S= \bigcup\limits_{i=1}^n S_i$.
\item We puncture $\widehat{H}$ on columns indexed by $S$. Let $\hhs$ be the resulting\\ $m(n-k) \times  \lambda n$ parity check matrix and let $\widehat{\cC}_S$ be the shortened code.
\end{itemize}
\paragraph{Hiding $\widehat{\cC}_S$}
\begin{itemize} 
\item Choose $n$ random $ \lambda \times  \lambda$ invertible matrices $T_1,T_2,\ldots,T_{n}$ over $\vF_{q}$. Define $T$ to be the block diagonal matrix having $T_1,T_2,\ldots, T_{n}$ as diagonal blocks.
\item Now choose a random permutation $\sigma$ of length $n$ and define $P_\sigma$ to be the block permutation matrix of size $ \lambda n \times  \lambda n$. It can also be seen as Kronecker product of the $n \times n$ permutation matrix corresponding to $\sigma$ and the identity matrix of size $ \lambda$.
\item  Define $Q := TP_\sigma$ and compute $H'= \hhs Q$.
\end{itemize}

The private key is then $(H,Q,\gamma)$ and the public key is $(H', t, \lambda)$.

\paragraph{\textbf{Encryption:}} Let $y \in \vF_{q}^{ \lambda n}$ be a message having support in $t$ sub-vectors each of length $ \lambda$, in particular
\begin{align*}
\text{support}(y) \subseteq \left \lbrace \lambda \right.& (i_1-1)+1, \lambda (i_1-1)+2,  \ldots, \lambda (i_1),  \lambda (i_2-1)+1,  \lambda (i_2-1)+2, \\ & \ldots, \lambda (i_2),\ldots,\left. \lambda (i_t-1)+1,\lambda (i_t-1)+2,\ldots, \lambda (i_t) \right \rbrace, 
\end{align*} for some distinct $i_1,i_2,\ldots,i_t \in \lbrace1,2,\ldots,n \rbrace$.  Then compute  the cipher text
\begin{equation*}
c = H' y^\intercal.
\end{equation*}

\paragraph{\textbf{Decryption:}} For the decryption we apply $\phi_{n-k}^{-1}$ on $c$, i.e.
\begin{eqnarray*}
\phi_{n-k}^{-1}(c) & = & \phi_{n-k}^{-1}\left(\hhs Qy^\intercal \right). 
\end{eqnarray*}
Observe that $\hhs Qy^\intercal  = \widehat{H} \bar{y}^\intercal$, where $\bar{y}$ is the embedding of $yQ^\intercal $ to $\vF_{q^m}$, by introducing zeros on the positions indexed by $S$. From  Proposition \ref{prop:phipsi} we get
\begin{eqnarray*}
\phi_{n-k}^{-1}\left(\widehat{H} \bar{y}^\intercal \right)  & = & H \left(\phi_n^{-1}(\bar{y})\right)^\intercal.
\end{eqnarray*}

Due to the block structure of the matrix $Q$, the vector of $Qy^\intercal$ has support in $t$ sub-vectors each of length $\lambda$, thus $ \bar{y}$ has support in $t$ sub-vectors each of length $m$.   Henceforth $\text{wt}(\phi_n^{-1}(\bar{y})) \leq  t $, and we can decode $\phi_{n-k}^{-1}(c)$ to get $\phi_n^{-1}(\bar{y})$. By applying $\phi_n$ we get $\bar{y}$ and by projecting on positions not indexed by $S$, we get $Qy^\intercal$  and therefore after multiplying by $Q^{-1}$, we recover the message $y$.

\subsection*{Choice of parameters}

For low key sizes it is desirable to use a small degree of extension $m$ and small $\lambda$. 

In the case of quadratic extension and in the case of $\lambda=1$, puncturing all but one column from each block results in an alternant code (subfield subcode of a GRS code). Alternant codes are known to be vulnerable to square code attacks \cite{co17w,fa13}. Hence, we do not propose to use quadratic extensions or $\lambda=1$. 

We therefore propose to use $m=3$ and $m=4$ with $\lambda= 2$.

\section{Security}\label{sec:Security}

In this section we  discuss the security of the proposed cryptosystem. 
We  focus on the three main attacks on cryptosystems based on GRS codes. Two  of them are structural (or key recovery) attacks, namely the Sidelnikov-Shestakov attack and the distinguisher attack based on the Schur product of the public code. The third one is the best known non-structural attack called information set decoding (ISD).

\subsection{Sidelnikov and Shestakov attack }
The first code-based cryptosystem using GRS codes as secret codes was proposed by Niederreiter in the same article \cite{ni86a} as the famous Niederreiter cryptosystem. This proposal was then attacked by Sidelnikov and Shestakov in \cite{si92a}, where they used the fact, that the public matrix is still a generator matrix of a GRS code and they were able to recover the evaluation points and hence the GRS structure of the public matrix. 

In the  cryptosystem proposed in Section \ref{sec:cryptosystem}, the secret GRS parity check matrix $H$ over $\vF_{q^m}$ is hidden in two ways: first by puncturing its expanded parity check matrix $\widehat{H}$ over $\vF_q$ and then by scrambling the columns of the punctured matrix $\hhs$. Due to multiplying $\hhs$ with a block diagonal matrix it is clear that the resulting code is no more equivalent to an evaluation code (or an expanded evaluation code).
Hence evaluations (or expanded evaluation column vectors) can not be exploited using the Sidelnikov-Shestakov attack.

\subsection{Distinguisher attack based on the Schur product}

For the attack based on the Schur product we need to introduce some definitions and notations. 

\begin{definition}[Schur product]
Let $x,y \in \vF_q^n$. We denote by the Schur product of $x$ and $y$ their component-wise product
\begin{equation*}
x \star y = (x_1  y_1, \ldots, x_n y_n).
\end{equation*}
\end{definition}

\begin{remark}
The Schur product is symmetric and bilinear. \label{remark_schur}
\end{remark}
\begin{definition}[Schur product of codes and square code]

Let $\mathcal{A},\mathcal{B}$ be two codes of length $n$. The Schur product of two codes is the vector space spanned by all $a \star b$ with $a \in \mathcal{A}$ and $b \in \mathcal{B}$:
\begin{equation*}
\langle \mathcal{A} \star \mathcal{B} \rangle  = \langle \{ a \star b \bigm| a \in \mathcal{A}, b\in \mathcal{B} \} \rangle.
\end{equation*}
If $\mathcal{A} = \mathcal{B}$, then we call $ \langle \mathcal{A} \star \mathcal{A} \rangle $ the square code of $\mathcal{A}$ and denote it by $ \langle \mathcal{A}^2 \rangle $.
\end{definition}

\begin{definition}[Schur matrix]
Let $G$ be a $k \times n$ matrix, with rows $(g_i)_{1 \leq i \leq k}$.  The Schur matrix of $G$, denoted by $S(G)$,  consists of the rows
$
g_i \star g_j $
for \\ $1 \leq i \leq j \leq k.$
\end{definition}
We observe by Remark \ref{remark_schur}, that if $G$ is a generator matrix of a code $\cC$ then its Schur matrix $S(G)$ is a generator matrix of the square code of $\cC$. Let $s$ be the following map 
\begin{eqnarray*}
s: \mathbb{N} & \to &  \mathbb{N} \\
k & \mapsto & \dfrac{1}{2}\left(k^2+k\right).
\end{eqnarray*}
For a $k \times n$ matrix $A$, we observe that $S(A)$ has the size $s(k) \times n$. \\

Various McEliece cryptosystems based on modifications of GRS codes have been proved to be insecure \cite{co14,co15a,ga12}. This is because the dimension of the square code of GRS codes is very low compared to a random linear code of the same dimension.
Moreover, other families of codes have also been shown to be vulnerable against the attacks based on Schur products. In \cite{co17}, Couvreur \textit{ et al.} presented a general attack against cryptosystems based on algebraic geometric codes and their subcodes. In \cite{fa13} Faug\`{e}re \textit{et al.} showed that high rate binary Goppa codes can be distinguished from a random code. 
In \cite{co17w}, Couvreur \textit{et al.} presented a polynomial time attack against cryptosystems based on non-binary Goppa codes defined over quadratic extensions.

The distinguisher attack is based on the low dimensional square code of the public code (or of the shortened public code). In the following, based on experimental observations, we infer that the public code of the proposed cryptosystem cannot be distinguished using square code techniques. 

Let $\widehat{\cC}_S$ be the public code of the proposed cryptosystem. Note that $\ccs$ is a shortening of an expanded GRS code $\widehat{\cC}$. 

\begin{enumerate}
\item \emph{Squares of expanded GRS codes}: Like in the case of Reed-Solomon codes and their subfield subcodes, the expanded GRS codes also have low square code dimension. To see this, we visualize expanded GRS codes as subfield subcodes of GRS-like codes. Let $\cC$ be a GRS code of length $n$ and dimension $k$ over $\vF_{q^m}$ having the following parity check matrix  \[H = V_{r}(x,y) := \begin{pmatrix}
y_1 & y_2 & \cdots & y_n \\
y_1 x_1 & y_2 x_2 & \cdots & y_n x_n \\
\vdots & \vdots & \ddots & \vdots \\
y_1 x_1^{r-1} & y_2 x_2^{r-1} & \cdots & y_n x_n^{r-1} \\
\end{pmatrix},\] where $x = (x_1,\ldots,x_n)$ is a vector of distinct elements in $\vF_{q^m}$, $y = (y_1,\ldots,y_n)$ is a vector over $\vF_{q^m}^\ast$ and $r := n-k$. Let $\gamma$ be a primitive element in $\vF_{q^m}$. We define a new code $\mathcal{B}$ of length $mn$ over $\vF_{q^m}$ given by the kernel of the following parity check matrix \[H^\prime = \begin{pmatrix}
V_{r}(x,y) & \mid & V_r(x,\gamma y) & \mid & \cdots & \mid & V_r(x,\gamma^{m-1} y)
\end{pmatrix}.\] Using Lemma \ref{lemma:expand}, it is easy to observe, that the expanded code $\widehat{\cC}$ of $\cC$ with respect to $\gamma$ is permutation equivalent to the $\vF_q$-kernel of $H^\prime$. In other words $\widehat{\cC}$ is permutation equivalent to the subfield subcode of $\mathcal{B}$ over $\vF_q$. Observe that a generator matrix $G^\prime $ of $\mathcal{B}$ is given by 
\begin{equation*}
\begin{pmatrix}
V_k(x,y^\prime) & 0 & \ldots & 0 & 0 \\
0 & V_k(x,\gamma^{-1} y^\prime) & \ldots & 0 & 0 \\
\vdots & \vdots &  \ddots & \vdots & \vdots \\
0 & 0 & \ldots & V_k(x,\gamma^{-(m-2)}y^\prime) & 0 \\
0 & 0 & \ldots & 0 & V_k(x,\gamma^{-(m-1)}y^\prime )\\ \hline
V_r(x,y^{\prime \prime}) & 0 & \ldots & 0 & - V_r(x,\gamma^{-(m-1)} y^{\prime \prime}) \\
 0 & V_r(x,\gamma^{-1} y^{\prime \prime}) &  \ldots & 0 & - V_r(x, \gamma^{1-(m-1)} y^{\prime \prime}) \\
\vdots & \vdots & \ddots & \vdots & \vdots \\
0 & 0  & \ldots & V_r(x, \gamma^{-(m-2)} y^{\prime \prime}) & - V_r(x,\gamma^{(m-2)-(m-1)} y^{\prime \prime}) \\
\end{pmatrix},
\end{equation*}
where $y^\prime$ is such that $V_k(x,y^\prime) V_r(x,y)^\intercal = 0$, and $y^{\prime \prime} = (x_1^k,x_2^k,\ldots,x_n^k)\star y^\prime$. One can verify that $G^\prime (H^\prime)^\intercal = 0$. 
Observe that a generator matrix of $\widehat{\cC}$ is permutation equivalent to \[\widehat{G} = \begin{pmatrix}
G_1 & 0 & \ldots & 0 \\
0 & G_2 & \ldots & 0 \\
\vdots & \vdots & \ddots & \vdots \\
0 & 0 & \ldots & G_m \\ \hline
\multicolumn{4}{c}{G_{gv}} \\
\end{pmatrix}, \]
where $G_i$ is a generator matrix of the subfield subcode of $V_k(x,\gamma^{-(i-1)}y^\prime)$ over $\vF_q$, and $G_{gv}$ is a generator matrix of the $\vF_q$-subfield subcode of the bottom $(m-1)r$ rows of $G^\prime$. The matrix $G_{gv}$ is also known as the glue-vector generator matrix, as in \cite{va91}. Due to the block structure of $\widehat{G}$ the Schur matrix of $\widehat{G}$ will have many zero rows. As a result the dimension of the square code is not full, given large enough $n$. This may lead to vulnerabilities when using expanded GRS codes directly in the cryptosystem. 
\item \emph{Effect of Shortening}: Consider the parity check matrix $\widehat{H}$ of an expanded GRS code as shown in Lemma \ref{lemma:expand}. We partition the columns of $\widehat{H}$ into $n$ blocks, each of size  $m$. By the definition of $\widehat{H}$, each of these blocks corresponds to a unique column vector of the parity check matrix of the parent GRS code. In order to weaken this correspondence, we puncture (randomly chosen) $m -\lambda$ of the columns from each block of $\widehat{H}$. As a result the correspondence of each block to the parent column vector is inconsistent. In addition we multiply the punctured parity check matrix by an invertible block diagonal matrix $T$. This further destroys the algebraic structure inherited from the parent GRS code. 
This was evident in our computations of the square code dimension of such shortened codes. Even in the case of $m=3$ we observed that puncturing one column from each block of $\widehat{H}$ results in a full square code dimension. 
\end{enumerate}



\subsection{Information Set Decoding}\label{sec:ISD}

Information set decoding (ISD) algorithms are the best known algorithms for decoding a general linear code. ISD algorithms were introduced by Prange \cite{pr62} in 1962. Since then several improvements have been proposed for codes over the binary field by Lee-Brickel \cite{le88}, Leon \cite{le88a}, Stern \cite{st89} and more recently by Bernstein \textit{et al.}  \cite{be11}, Becker \textit{et al.} \cite{be12}, May-Ozerov \cite{ma15}. Several of these algorithms have been generalized to the case of codes over general finite fields, see \cite{kl17,hi16,in18,ni17,pe10}.   

An ISD algorithm in its simplest form  first chooses an information set $I$, which is a size $k$ subset of $\lbrace 1,2, \ldots,n \rbrace$ such that the restriction of the parity check matrix on the columns indexed by the complement of $I$ is non-singular. Then Gaussian elimination brings the parity check matrix in a standard form and assuming that the errors are outside of the information set, these row operations on the syndrome will exploit the error vector, if the weight does not exceed the given error correction capacity. 

\paragraph{ISD for the proposed cryptosystem:}
In the proposed cryptosystem we introduce a burst pattern in the  error vector, in particular the error vector has support in $t$ sub-vectors each of length $\lambda$. Henceforth, we modify Stern's ISD algorithm to incorporate such pattern in the error vector.

We first recall the Stern’s algorithm. The algorithm partitions the information set $I$ into two equal-sized subsets $X$ and $Y$, and chooses uniformly at random a subset $Z$ of size $\ell$ outside of  $I$. Then it looks for vectors having exactly weight $p$ among the columns indexed by $X$, exactly weight $p$ among the columns indexed by $Y$, and exactly weight 0 in columns indexed by $Z$ and the missing weight $t-2p$ in the remaining indices.

In the proposed cryptosystem we have been given a public code $\ccs$ of length $\lambda n$ and dimension $k^\prime := mk-(m-\lambda)n$ over $\vF_q$. We also know that the error vector has support in $t$ sub-vectors of length $\lambda$. Hence we use Stern's algorithm on the blocks of size $\lambda$. We consider the information set $I$ to have $\left \lfloor k^\prime/\lambda \right \rfloor$ blocks. We partition $I$ into two equal-sized subsets $X$ and $Y$, and choose uniformly at random a subset $Z$ of $\ell$ blocks outside of $I$. Then we look for vectors  having support in exactly $p$ blocks in $X$, exactly $p$ blocks in $Y$, and exactly 0 blocks in $Z$.

In Section \ref{sec:KeySize} we compute the  key sizes of the proposed cryptosystem having 256-bit security against this modified ISD algorithm.

\section{Key size}\label{sec:KeySize}

In this section we compute the key sizes of the proposed cryptosystem having 256-bit security against the ISD algorithm discussed in Section \ref{sec:ISD}. Later we compare these key sizes with the key sizes of the  McEliece cryptosystem using binary Goppa codes \cite{be08} and some recently proposed cryptosystems that are using Reed-Solomon codes as secret codes. These are based on the idea of  \cite{ba11,ba15z} (BBCRS), where the authors proposed to hide the structure of the code using as transformation matrix the sum of a  rank $z$ matrix and a weight $w$ matrix. The proposed parameters in \cite{ba11,ba15z} with $z=1$ and $w\leq 1+R$ were broken by the square code attack \cite{co14,co15a}, where $R$ denotes the rate of the code. 
Two countermeasures were recently proposed in \cite{ba19,kh18}.
In order to hide the structure of the Reed-Solomon code  the authors of \cite{ba19} use $w>1+R$ and $z=1$ or $w<1+R$ and $z>1$. Whereas in \cite{kh18} the transformation matrix has weight $w=2$ and rank $z=0$.

In the proposed cryptosystem, the public key is a parity check matrix of a linear code over $\vF_q$ having length $\lambda n$ and dimension $mk-(m-\lambda) n$. Hence the public key size is $(\lambda n-m(n-k)) \cdot m(n-k) \cdot \log_2(q)$ bits. For a degree of extension $m$, let $\cC_m$ be the public code. 

In Table \ref{table:m2KS}, we provide the key sizes for different rates of the public code $\cC_3$ achieving a 256-bit security level against the modified ISD algorithm discussed in Section \ref{sec:ISD}. Observe that the smallest key size is achieved at rate $0.82$. 
\begin{table}[h!]
\begin{center}
\begin{tabular}{c|c|c|c|c|c}
	Rate & $q$ & $n$ & $k$ & $t$ &  Key Size (bits) \\
	\hline
	0.60 & 13 & 1382 & 829 & 277 & 6783627 \\
	0.65 & 13 & 1270  & 825 & 223 & 5952804 \\
	0.70 & 13 & 1207 & 844 & 182 & 5339456 \\
	0.75 & 13 & 1192 & 894 & 149 & 4929077 \\
	0.80 & 13 & 1230 & 984 & 123 & 4702652 \\
	0.82 & 13 & 1258 & 1031 & 114 & 4624198 \\
	0.85 & 13 & 1340 & 1139 & 101 & 4634545 \\ 
	0.87 & 13 & 1420 & 1235 & 93 & 4692805 \\ 
	0.90 & 13 & 1602 & 1441 & 81 & 4863276 \\ \hline
\end{tabular}
\caption{Comparing key sizes of the proposed cryptosystem with $m=3$ and $\lambda=2$ reaching a $256$-bit security level against the modified ISD algorithm.} \label{table:m2KS}
\end{center}
\end{table}

In Table \ref{table:m4KS}, we provide the key sizes for different rates of the public code $\cC_4$ achieving a 256-bit security level against the modified ISD algorithm discussed in Section \ref{sec:ISD}. In this case the smallest key size is achieved at rate $0.89$. 

\begin{table}[h!]
\begin{center}
\begin{tabular}{c|c|c|c|c|c}
	Rate & $q$ & $n$ & $k$ & $t$ & Key Size (bits) \\
	\hline
	0.65 & 7 & 2360 & 1534 & 413 & 13134108 \\
	0.70 & 7 & 1945 & 1361 & 292 & 10191102 \\
	0.75 & 7 & 1738 & 1303 & 218 & 8480009 \\
	0.80 & 7 & 1662 & 1329 & 167 & 7448878 \\
	0.85 & 7 & 1700 & 1445 & 128 & 6815134 \\
	0.87 & 7 & 1770 & 1539 & 116 & 6785893 \\
	0.89 & 7 & 1872 & 1666 & 103 & 6754721 \\
	0.91 & 7 & 2024 & 1841 & 92 & 6814326 \\ \hline
\end{tabular}
\caption{Comparing key sizes of the proposed cryptosystem with $m=4$ and $\lambda=2$ reaching a $256$-bit security level against the modified ISD algorithm.} \label{table:m4KS}
\end{center}
\end{table}

In conclusion, for a $256$ bit security level we propose to use the cryptosystem with the two sets of parameters $(q=13, m=3, \lambda=2, n=1258 ,k= 1031)$ and $(q=7, m=4, \lambda=2, n=1872, k=1666)$, see Table \ref{Table:proposed}.
\begin{table}[h!]
\begin{center}
\begin{tabular}{ll|c|c|c|c |c}
	  &  & $q$ & $m$ & $n$ & $k$ &  Key Size (in bits) \\
	\hline
	\multirow{2}{1.8cm}{Proposed system } & Type I  &  13 & 3 & 1258 & 1031 & 4624198 \\
	& Type II  & 7  & 4 & 1872 & 1666  & 6754721 \\ \hline
	\multicolumn{2}{l|}{classical McEliece} & 2 & 13 & 6960 & 5413 & 8373911 \\ \hline
	\multirow{3}{1.8cm}{BBCRS based schemes} & $w=1.708$ and $z=1$ & 1423 & 1 & 1422 & 786 & 5113520 \\
    & $w=1.2$ and  $z=10$ & 1163 & 1 & 1162 & 928 & 2274160 \\
    & $w=2$ and $z=0$ & 1993 & 1 &1992 & 1593 & 6966714  \\  \hline
\end{tabular} 
\caption{Comparing the key sizes of the proposed parameters against different cryptosystems.} \label{Table:proposed}
\end{center}
\end{table}

 
The proposed parameters for the classic McEliece system using binary Goppa codes by Bernstein \textit{et al.} in \cite{be08} are $q=2, m=13, n= 6960, k = 5413$, which gives a key size of $8373911$ bits. It achieves a security level of 260-bits with respect to the ball-collision algorithm \cite{be11}. 

In comparison to the classic McEliece system, the Type I set of parameters reduces the key size  by $44.8 \%$ and the Type II set of parameters reduces the key size  by $ 19.3\%$.


\section{Acknowledgement}
The authors would like to thank Matthieu Lequesne and Jean-Pierre Tillich for  pointing out the square code vulnerability in the case of quadratic extensions. This work has been supported by the Swiss National Science Foundation under grant no. 169510.

\end{document}